\documentclass{article}
\usepackage{ijcai19}

\usepackage{times}
\usepackage{soul}
\usepackage{url}
\usepackage[hidelinks]{hyperref}
\usepackage[utf8]{inputenc}
\usepackage[small]{caption}
\usepackage{graphicx}
\usepackage{amsmath}
\usepackage{booktabs}
\usepackage{multirow}
\usepackage{amsfonts, amssymb, mathrsfs}
\usepackage{amsthm}
\usepackage[ruled,linesnumbered,noresetcount,vlined]{algorithm2e}

\newcommand{\scope}[1]{\mathbf{x}^{#1}}
\newcommand{\setf}[1]{{\bf{#1}}}

\newcommand{\argmax}{\operatornamewithlimits{argmax}}

\def\is{\!=\!}

\newtheorem{theorem}{{\bf Theorem}}

\newcommand{\citet}[1]{\citeauthor{#1} [\citeyear{#1}]}

\newcommand{\bitemize}{\begin{list}{$\bullet$}{\topsep=1pt \parsep=0pt \itemsep=1pt \leftmargin=1em }} 
\newcommand{\eitemize}{\end{list}}
\newcommand{\beitemize}{\begin{list}{$\bullet$}{\topsep=1.5pt \parsep=0pt \itemsep=1pt \leftmargin=1em }} 
\newcommand{\enitemize}{\end{list}}

\theoremstyle{definition}

\title{New Algorithms for Functional Distributed Constraint Optimization Problems}

\author{
Khoi D. Hoang$^1$
\and
William Yeoh$^1$\and
Makoto Yokoo$^{2}$\And
Zinovi Rabinovich$^3$
\affiliations
$^1$Washington University in St. Louis, USA\\
$^2$Kyushu University, Japan\\
$^3$Nanyang Technological University, Singapore
\emails
\{khoi.hoang, wyeoh\}@wustl,edu,
yokoo@inf.kyushu-u.ac.jp,
zinovi@ntu.edu.sg
}

\allowdisplaybreaks

\begin{document}

\maketitle

\begin{abstract}
\small
The \emph{Distributed Constraint Optimization Problem} (DCOP) formulation is a powerful tool to model multi-agent coordination problems that are distributed by nature. The formulation is suitable for problems where variables are discrete and constraint utilities are represented in tabular form. However, many real-world applications have variables that are continuous and tabular forms thus cannot accurately represent constraint utilities. To overcome this limitation, researchers have proposed the \emph{Functional DCOP} (F-DCOP) model, which are DCOPs with continuous variables. But existing approaches usually come with some restrictions on the form of constraint utilities and are without quality guarantees. Therefore, in this paper, we (\emph{i}) propose exact algorithms to solve a specific subclass of F-DCOPs; (\emph{ii}) propose approximation methods with quality guarantees to solve general F-DCOPs; and (\emph{iii}) empirically show that our algorithms outperform existing state-of-the-art F-DCOP algorithms on randomly generated instances when given the same communication limitations.
\end{abstract}

\maketitle


\section{Introduction}
\label{sec:introduction}
The \emph{Distributed Constraint Optimization Problem} (DCOP)~\cite{modi:05,petcu:05} formulation is a powerful tool to model cooperative multi-agent problems.
DCOPs are well-suited to model many problems that are distributed by nature and where agents need to coordinate their value assignments to maximize the aggregate constraint utilities. 
This model is widely employed to model distributed problems such as meeting scheduling problems~\cite{maheswaran:04a}, sensor and wireless networks~\cite{farinelli:08,yeoh:12}, multi-robot teams coordination~\cite{zivan:15}, smart grids~\cite{kumar:09,miller:12,fioretto:17b}, coalition structure generation~\cite{ueda:10} and smart homes~\cite{rust:16,fioretto:17a}.

However, the regular DCOP model assumes that the variables are discrete and the constraint utilities are represented in tabular form (i.e.,~a utility is defined for every combination of discrete values of variables). While these assumptions are reasonable in some applications where values of variables correspond to a set of \emph{discrete} possibilities (e.g.,~the set of tasks that robots can perform in multi-robot coordination problems or the set of coalitions that agents can join in coalition structure generation problems), they make less sense in applications where values of variables correspond to a \emph{continuous} range of possibilities (e.g.,~the range of orientations a sensor can take in sensor networks or the range of frequencies an agent can choose in wireless networks). 

These limiting assumptions have prompted \citet{stranders:09a} to extend the DCOP formulation to allow for continuous variables. We refer to this extension as \emph{Functional DCOPs} (F-DCOPs) in this paper.\footnote{As \citet{stranders:09a} did not name their extension in their paper, we choose a name so that we can refer to it easily.} Additionally, as variables can now take values from a continuous range, constraint utilities are also similarly extended from tabular forms to functional forms in F-DCOPs. To solve such problems, \citet{stranders:09a} extended the discrete \emph{Max-Sum} (MS) algorithm~\cite{farinelli:08} to \emph{Continuous MS} (CMS), where constraint utility functions are approximated by piecewise linear functions. \citet{voice:10} later proposed \emph{Hybrid CMS} (HCMS), which combines the discrete MS algorithm with continuous non-linear optimization methods. Specifically, agents in HCMS approximate the utility functions with a number of samples that they iteratively improve over time. A key limitation of CMS and HCMS is that they both do not provide quality guarantees on the solutions found. The reason for this is that they rely on discrete MS as the underlying algorithmic framework, which do not provide quality guarantees on general graphs. 

To overcome this limitation, we extend the \emph{Distributed Pseudo-tree Optimization Procedure} (DPOP)~\cite{petcu:05} algorithm to three extensions -- \emph{Exact Functional DPOP} (EF-DPOP); \emph{Approximate Functional DPOP} (AF-DPOP); and \emph{Clustered AF-DPOP} (CAF-DPOP). EF-DPOP provides an exact approach to solve F-DCOPs with linear or quadratic utility functions and are defined over tree-structure graphs. Both AF-DPOP and CAF-DPOP solve F-DCOPs approximately without any restriction on the type of utility functions or graph structure. We also provide theoretical properties on the error bounds and communication complexities of AF-DPOP and CAF-DPOP and show that they outperform HCMS in randomly generated instances when given the same communication limitations. 


\section{Background}
\label{sec:background}


\paragraph{DCOPs:}

\noindent A \emph{Distributed Constraint Optimization Problem} (\emph{DCOP}) is a tuple  
$\langle \setf A, \setf X, \setf D, \setf F, \alpha\rangle$: $\setf{A} = \{a_i\}_{i=1}^p$ is a set of \emph{agents}; $\setf{X} = \{x_i\}_{i=1}^n$ is a set of decision \emph{variables}; $\setf{D} = \{D_x\}_{x \in \setf{X}}$ is a set of finite \emph{domains} and each variable $x \in \setf{X}$ takes values from the set $D_{x}$; $\setf{F} = \{f_i\}_{i=1}^m$ is a set of \emph{utility functions}, each defined over a set of decision variables: $ f_i : \prod_{x \in \scope{f_i}} D_x \to \mathbb{R} \cup \{-\infty\}$, where infeasible configurations have $-\infty$ utilities, $\scope{f_i} \subseteq \setf{X}$ is the \emph{scope} of $f_i$, and $\alpha : \setf{X} \to \setf{A}$ is a \emph{mapping function} that associates each decision variable to one agent. 

A \emph{solution} $\sigma$ is a value assignment for a set $\setf{x}_{\sigma} \subseteq \setf{X}$ of variables that is consistent with their respective domains. The utility $\setf{F}(\sigma) = \sum_{f \in \setf{F}, \scope{f} \subseteq \setf{x}_{\sigma}} f(\sigma)$ is the sum of the utilities across all the applicable utility functions in $\sigma$. A solution $\sigma$ is \emph{complete} if $\setf{x}_{\sigma} \is \setf{X}$. The goal is to find an optimal complete solution $\setf{x}^* = \argmax_{\setf{x}} \setf{F}(\setf{x})$.

A \emph{constraint graph} visualizes a DCOP, where nodes in the graph correspond to variables in the DCOP and edges connect pairs of variables appearing in the same utility function. A \emph{pseudo-tree} arrangement has the same nodes and edges as the constraint graph and satisfies that (\emph{i}) there is a subset of edges, called \emph{tree edges}, that form a rooted tree and (\emph{ii}) two variables in a utility function appear in the same branch of that tree. The other edges are called \emph{backedges}. Tree edges connect parent-child nodes, while backedges connect a node with its \emph{pseudo-parents} and its \emph{pseudo-children}. 


In this paper, we assume that each agent controls exactly one decision variable and thus use the terms ``agent'' and ``variable'' interchangeably. We also assume that all utility functions are binary functions between two variables. 

\paragraph{DPOP:}

\emph{Distributed Pseudo-tree Optimization Procedure (DPOP)}~\cite{petcu:05} is a complete \emph{inference algorithm} that is composed of three phases: 
\bitemize
\item \emph{Pseudo-tree Generation:} In this phase, all agents start building a pseudo-tree~\cite{hamadi:98}. 
\item \emph{UTIL Propagation:} Each agent, starting from the leaves of the pseudo-tree, adds the optimal sum of utilities in its subtree for each value combination of variables in its separator.\footnote{\small The separator of $x_i$ contains all ancestors of $x_i$ in the pseudo-tree that are connected to $x_i$ or to one of its descendants.} It does so by \emph{adding} the utilities of its functions with the variables in its separator and the utilities in the UTIL messages received from its children. The agent then \emph{projects} out its variable by optimizing over it and sends the projected function in a UTIL message to its parent.
\item \emph{VALUE Propagation:} Each agent, starting from the root of the pseudo-tree, determines the optimal value for its variable and then sends the optimal value as well as the optimal values of agents in its separator to its children. The root agent does so by choosing the values of its variables from its UTIL computations, and send them as VALUE messages.
\eitemize

\section{Functional DCOP Model}
\label{sec:model}

The \emph{Functional DCOP} (F-DCOP) model generalizes the regular discrete DCOP model by modeling the  variables as continuous decision variables~\cite{stranders:09a}. More formally, an F-DCOP is a tuple $\langle \setf A, \setf X, \setf D, \setf F, \alpha\rangle$, where $\setf{A}$, $\setf{F}$, and $\alpha$ are exactly as defined in DCOPs. The key differences are as follows: 
\beitemize
\item $\setf{X} = \{x_i\}_{i=1}^n$ is now a set of \emph{continuous} decision variables controlled by the agents.

\item $\setf{D} = \{D_x\}_{x \in \setf{X}}$ is now a set of \emph{continuous domains} of the decision variables. Each variable $x \in \setf{X}$ takes values from the interval $D_{x} = [LB_x, UB_x]$. 


\enitemize
The objective of an F-DCOP is the same as that in DCOPs -- to find an optimal complete solution $\setf{x}^* = \argmax_{\setf{x}} \setf{F}(\setf{x})$.



\section{F-DCOP Algorithms}
\label{sec:algorithm}

We now introduce three F-DCOP algorithms: \emph{Exact Functional DPOP} (EF-DPOP), \emph{Approximate Functional DPOP} (AF-DPOP), and \emph{Clustered AF-DPOP} (CAF-DPOP). All three algorithms are based on the framework of DPOP, where they  extend the capability of DPOP such that they can solve F-DCOPs with continuous variables and utility functions.


\subsection{Exact Functional DPOP}

\emph{Exact Functional DPOP} (EF-DPOP) is an exact algorithm for F-DCOPs with linear or quadratic utility functions and are defined over tree-structure graphs. It extends the two primary operations of DPOP in the UTIL propagation phase -- \emph{addition} and \emph{projection}.



\paragraph{Addition Operation:}
In EF-DPOP, each UTIL message contains a piecewise function 
and the addition of two piecewise functions is done by adding their sub-functions that may have different domains. We will use the following two functions to illustrate our operations:
\begin{small}\begin{align}
f_{12}(x_1, x_2) = 
\begin{cases}
f_{12}^a \qquad \text{if } x_1 \in [0,4], x_2 \in [0,6] \\
f_{12}^b \qquad \text{if }  x_1 \in [0,4], x_2 \in [6,10] \\
f_{12}^c \qquad \text{if }   x_1 \in [4,10], x_2 \in [0,6] \\
f_{12}^d \qquad \text{if }  x_1 \in [4,10], x_2 \in [6,10] \\
\end{cases} \\
f_{23}(x_2, x_3) = 
\begin{cases}
f_{23}^a \qquad \text{if } x_2 \in [0,3], x_3 \in [0,7] \\
f_{23}^b \qquad \text{if } x_2 \in [0,3], x_3 \in [7,10] \\
f_{23}^c \qquad \text{if } x_2 \in [3,10], x_3 \in [0,7] \\
f_{23}^d \qquad \text{if } x_2 \in [3,10], x_3 \in [7,10] \\
\end{cases}
\end{align}\end{small}

When adding two piecewise functions, we first identify the common variable between the two functions and create a new set of atomic ranges for the variable. For example, when adding the functions $f_{12}$ and $f_{23}$ above, the common variable is $x_2$, and the new ranges for $x_2$ are $[0, 3]$, $[3, 6]$, and $[6, 10]$. The ranges of the other variables remain unchanged from their original functions. 

We then take the Cartesian product of the range sets of all common variables and associate the appropriate function to that range. For example, the addition of $f_{12}$ and $f_{23}$ will be a new function $f_{123}$:
\begin{small}\begin{align}
f&_{123}(x_1, x_2, x_3) \notag \\ &=
\begin{cases}
f_{12}^a + f_{23}^a \quad \text{if } x_1 \in [0,4], x_2 \in [0,3], x_3 \in [0,7] \\
f_{12}^a + f_{23}^b \quad \text{if } x_1 \in [0,4], x_2 \in [0,3], x_3 \in [7,10] \\
f_{12}^c + f_{23}^a \quad \text{if } x_1 \in [4,10], x_2 \in [0,3], x_3 \in [0,7] \\
f_{12}^c + f_{23}^b \quad \text{if } x_1 \in [4,10], x_2 \in [0,3], x_3 \in [7,10] \\
\ldots
\end{cases}
\end{align}\end{small}


\vspace{-1.5em}
\paragraph{Projection Operation:}
Projecting out a variable $x_i$ from a function $f(x_i, x_{i_1},\dots, x_{i_k})$ means finding the piecewise function:
\begin{small}\begin{align}
g(x_{i_1},\dots, x_{i_k}) = \argmax_{x_i} f(x_i, x_{i_1},\dots, x_{i_k}) \label{argmax-multi}
\end{align}\end{small}
To find $g$, we solve the following for closed-form solutions:
\begin{small}\begin{align}
\frac{\partial f(x_i, x_{i_1},\dots, x_{i_k})}{\partial x_i} &= 0 \label{partial-multi} 
\end{align}\end{small}
Let $\bar{x}_i$ be the solution to the equation above. Then:
\begin{small}\begin{align}
\bar{x}_i  &= g'(x_{i_1},\dots, x_{i_k}) \\ \label{root-partial-multi}
\bar{g}(x_{i_1},\dots, x_{i_k}) &= f(x_i = \bar{x}_i, x_{i_1},\dots, x_{i_k})
\end{align}\end{small}
Aside from $\bar{g}$, there are two other candidate functions:
\begin{small}\begin{align}
\check{g} &= f(x_i = LB_{x_i}, x_{i_1},\dots, x_{i_k}) \\
\hat{g} &= f(x_i = UB_{x_i}, x_{i_1},\dots, x_{i_k})
\end{align}\end{small}
Next, we need to find the intervals where each of the functions $\bar{g}, \check{g}$ and $\hat{g}$ is the largest. Those intervals are the intersections between the three functions and, thus, we solve each of the equations below to find them:
\begin{small}\begin{align}
\check{g} (x_{i_1},\dots, x_{i_k}) &= \hat{g} (x_{i_1},\dots, x_{i_k}) \\
\check{g} (x_{i_1},\dots, x_{i_k}) &= \bar{g} (x_{i_1},\dots, x_{i_k})\\
\hat{g} (x_{i_1},\dots, x_{i_k}) &= \bar{g} (x_{i_1},\dots, x_{i_k})
\end{align}\end{small}
The result of this process is a set of intervals where either $\bar{g}$, $\check{g}$, or $\hat{g}$ is the largest.
The projected function $g$ is the piecewise function that consists of $\bar{g}$, $\check{g}$, or $\hat{g}$ with the intervals that they are the largest in.

Unfortunately, it is not always possible to find closed-form solutions to the partial derivative in Equation~\eqref{partial-multi}. We discuss below two types of functions -- binary linear and quadratic functions -- where it is possible to find closed-form solutions. 
\beitemize
\item Binary linear functions of the form $f(x_i, x_{i_1}) = ax_{i} + bx_{i_1} + c$.
By following the monotonicity property of linear functions, we can find $g(x_{i_1}) = \argmax_{x_i} f(x_i, x_{i_1})$ at the two extremes:
\begin{small}\begin{align}
g(x_{i_1}) = 
\begin{cases} 
f(x_i = LB_{x_i}, x_{i_1}) &\text{if} \quad a > 0 \\
f(x_i = UB_{x_i}, x_{i_1}) &\text{otherwise}
\end{cases}
\end{align}\end{small}

\item Binary quadratic functions of the form $f(x_i, x_{i_1}) = a x_i^2 + bx_i + cx_{i_1}^2 + dx_{i_1} + ex_ix_{i_1} + f$. We first take the partial derivative and setting it to 0 to find the critical point:
\begin{small}\begin{align}
\frac{\partial f}{\partial x_i} = 
2ax_i + b + ex_{i_1}  &= 0\\
\bar{x}_{i} = \frac{-b - ex_{i_1}}{2a}
\end{align}\end{small}
As $\bar{x}_{i} $ has to belong to the interval $[LB_{x_i}, UB_{x_i}]$, we solve the inequalities below to find the range $x_{i_1}$ as the domain of $\bar{g}(x_{i_1})$:
\begin{small}\begin{align}
LB_{x_i} \le \frac{-b - ex_{i_1}}{2a} \le UB_{xi}
\end{align}\end{small}
\enitemize

\subsection{Approximate Functional DPOP}

In general F-DCOPs, Eq.~\eqref{partial-multi} may be a multivariate equation, and it is not always possible to find a closed-form solution to such functions.
Therefore, an approximation approach is desired for F-DCOPs.

In this section, we introduce \emph{Approximate Functional DPOP} (AF-DPOP), which is an approximation algorithm that can solve F-DCOPs without any restriction on the functional form of the constraint utilities. AF-DPOP is similar to DPOP in that the algorithm has the same three phases: pseudo-tree generation, UTIL propagation, and VALUE propagation phases. The pseudo-tree generation phase is identical to that of DPOP, and the UTIL and VALUE propagation phases share some similarities. 

We now describe how these two propagation phases work at a high level. In the UTIL propagation phase, like DPOP, agents in AF-DPOP also discretizes the domains of variables and sends up UTIL tables that contain utilities for each value combination of values of separator agents. However, unlike DPOP, agents in AF-DPOP perform local optimization of these values by ``moving'' them along the gradients of relevant utility functions in order to improve the overall solution quality. As such, the addition and projection operators have to be updated as well.

In the VALUE propagation phase, like DPOP, agents in AF-DPOP also sends down their best value down to their children in the pseudo-tree. However, unlike DPOP, agents in AF-DPOP may receive values of ancestors that do not map to computed utilities. As such, the agents must perform local interpolation of the utilities value in this phase.

We now describe the algorithm in more detail, where we focus on the UTIL and VALUE propagation phases of the algorithm. 


\paragraph{UTIL Propagation:} In this phase, each leaf agent first discretizes the domains of the agents in its separator (i.e.,~its parent and pseudo-parents) and then stores the Cartesian product of these discrete values in the set $V$. Therefore, each element $v \in V$ is a tuple $\langle v_{i_1}, \ldots, v_{i_k}\rangle$, where each value $v_{i_j}$ is the value of separator agent $a_{i_j}$. 

Then, for each tuple $v \in V$, the agent ``moves'' each value $v_{i_j}$ in the tuple along the gradient of each function that is relevant to agent $a_{i_j}$. Specifically, the agent updates value $v_{i_j}$ according to the following equation for each separator agent $x_{i_j}$ of the leaf agent $x_i$:
\begin{align}
v_{i_j} = v_{i_j} + \alpha \left.\frac{\partial f_{i_j}(x_i, x_{i_j})}{\partial x_{i_j}} \right\vert^{v_{i_j}}_{\argmax_{x_i} f_{i_j}(x_{i_j} = v_{i_j})}
\label{eq:leaf_move}
\end{align}
where $f_{i_j}(x_i, x_{i_j})$ is the utility function between the leaf agent $x_i$ and the separator agent $x_{i_j}$ and $\alpha$ is the \emph{learning rate} of the algorithm. The agent can ``move'' the values as many times as it like until they have either converged or a maximum number of iterations is reached. The agent can ``move'' the values as many times as desired until they have either converged or a maximum number of iterations is reached. Then, the updated values in $V$ and their corresponding utilities define the UTIL table that is sent to the parent of the agent in a UTIL message.

As in DPOP, each non-leaf agent will first wait for the UTIL messages from each of its children. When all the UTIL messages are received, the agent processes the UTIL tables in the UTIL message from each child. Note that in regular DPOP, the Cartesian product of the values of agents are consistent across the UTIL tables of all children (i.e., if the values of an agent $a$ exists in the Cartesian products of two children, then those values are identical). The reason is because all agents agree on the discretization of the domain of agent $a$ and do not update the value of that agent (such as through Eq.~\eqref{eq:leaf_move}). Therefore, each agent can easily add up the utilities in the UTIL tables received together with the utilities of constraints between the agent and its separator.

In contrast, since the values of agents are updated according to Eq.~\eqref{eq:leaf_move} in AF-DPOP, these values may no longer be consistent across different UTIL tables received. To remedy this issue, each agent first adds additional tuples to each UTIL table received such that the Cartesian product of the values of agents are consistent across all the UTIL tables. Then, it approximates the utilities of the newly added tuples by interpolating between the utilities of the existing tuples. Finally, since the UTIL tables are now all consistent, the agent adds up the utilities in the UTIL tables of children together with the utilities of constraints between the agent and its separator in the same way as DPOP.

After the utilities are added up, similar to leaf agents, the agent $x_i$ will proceed to repeatedly update the values $v_{i_j}$ of the separator $a_{i_j}$ in the updated Cartesian product $V$ using:
\begin{align}
v_{i_j} = v_{i_j} + \alpha \left.\frac{\partial f_{i_j}(x_i, x_{i_j})}{\partial x_{i_j}} \right\vert^{v_{i_j}}_{\argmax_{x_i} UTIL_i(v_{i_1}, \ldots, v_{i_k})}
\label{eq:nonleaf_move}
\end{align}
where $UTIL_i$ is the utility table that is constructed from the summation of the children's utilities and the utilities of constraints between the agent $x_i$ with its separator set.
The key difference between this Eq.~\eqref{eq:nonleaf_move} and the Eq.~\eqref{eq:leaf_move} used by leaf agents is that the substitution of $f_{i_j}(x_{i_j} = v_{i_j})$ with $UTIL_i(v_{i_1}, \ldots, v_{i_k})$. 

The reason for this substitution is that the utilities in the UTIL tables of leaf agents are only a function of constraints with their separator agents and the functional form of those constraints are known. Therefore, leaf agents can optimize exactly those functions to get accurate gradients. In contrast, utilities in the UTIL tables of non-leaf agents are also a function of the constraints between its descendant agents and its separator agent, and the functional form of those constraints are not known. They are only represented by samples within the UTIL tables received and are now integrated into the UTIL table of the non-leaf agent. Therefore, in Eq.~\eqref{eq:nonleaf_move}, the agent approximates its maximum value $x_i$ by choosing the best value of under the assumption that the values of the other separator agents are exactly the same as in the tuple $\langle v_{i_1}, \ldots, v_{i_k} \rangle$ that is being updated. 

After these values are all updated, the agent approximates their corresponding utilities by interpolating between known utilities and sends these utilities up to its parent in a UTIL message. These UTIL messages propagate up to the root agent, which then starts the VALUE phase.

\paragraph{VALUE Propagation:} The root agent starts this phase after processing all the UTIL messages received from its children in the UTIL phase. It chooses its best value based on its computed UTIL table and sends this value down to its children. Like in DPOP, each agent will repeat the same process after receiving the values of its parent and pseudo-parents. 

However, unlike DPOP, an agent may receive the information that its parent or pseudo-parent is taking on a value that doesn't correspond to an existing value in the agent's UTIL table due to the values being moved during the UTIL propagation phase. As a result, the agent will need to approximate the utility for this new value received and it does so by interpolating between known utilities in its UTIL table. 

Once all the leaf agents receive VALUE messages from their parents and choose their best values, the algorithm terminates.

%
	
\subsection{Clustered Approximate Functional DPOP}

A possible limitation of AF-DPOP is that the number of tuples in the Cartesian product $V$ that is propagated in the UTIL messages can be quite large, especially if additional tuples are added to maintain consistency between the UTIL tables of children. In communication-constrained applications, it is preferred that the number and size of messages transmitted between agents to be as small as possible.

With this motivation in mind, we extend AF-DPOP to \emph{Clustered AF-DPOP} (CAF-DPOP), which bounds the number of tuples sent in UTIL messages to limit the message size. CAF-DPOP is identical to AF-DPOP in every way except that agents choose $k$ representative tuples and their corresponding utilities to be sent up to their parents in UTIL messages. To choose these $k$ representative tuples, we use the $k$-means clustering algorithm~\cite{macqueen:1967} to cluster the tuples and then approximate the utilities of those tuples through interpolation. This approach assumes that tuples that are close to each other will have similar values. 

Note that while only $k$ tuples are sent between agents in UTIL messages, each agent still maintains the original unclustered set of tuples in their memory. Thus, when they perform interpolation during the VALUE propagation phase, they will use the utilities of the unclustered set of tuples since they are more accurate than the utilities of the clustered set of tuples.

\section{Theoretical Properties}

\begin{table*}[t] \small
	\begin{center}
		\small \centering
		\scalebox{1}
		{%
			\begin{tabular}{r | r | r | rrrr | r |}
				\cline{2-8}	
				\multicolumn{1}{c|}{\multirow{2}{*}{$|\setf{A}|$}} 
				& \multicolumn{1}{c|}{\multirow{2}{*}{HCMS}} & \multicolumn{1}{c|}{\multirow{2}{*}{DPOP}} & \multicolumn{4}{c|}{AF-DPOP} & \multicolumn{1}{c|}{\multirow{2}{*}{EF-DPOP}} \\
				& & & 5 & 10 & 15 & 20 &\\
				\cline{2-8}	
				10 & 129k & 220k & 330k & 356k & 374k & 404k & 518k \\
				20 & 306k & 541k & 795k &  870k & 947k & 1008k &  --- \\  		
				30 & 436k & 766k  & 1128k & 1230k & 1331k & 1414k & --- \\
				40 & 636k & 1104k & 1587k & 1728k & 1876k & 1980k  & 	--- \\
				50 & 832k & 1456k & 2109k & 2316k & 2533k & 2687k & --- \\
				\cline{2-8} 
			\end{tabular}%
		}
		\vspace{-0.5em}
		\caption{Experimental Results Varying the Number of Agents on Random Trees with Three Initial Discrete Points \label{table:tree-iteration}}
	\end{center}
	\vspace{-0.5em}
\end{table*}

We now provide some theoretical properties of some of our algorithms as well as that of (discrete) DPOP and Hybrid Continuous Max-Sum (HCMS). 

For each reward function $f(x_i, x_{i_1}, \dots, x_{i_k})$ of an agent $x_i$ and its separator agents $x_{i_1}, \ldots, x_{i_k}$, assume that agent $x_i$ discretizes the domains the reward function into hypercubes of size $m$ (i.e., the distance between two neighboring discrete points for the same agent $x_{i_j}$ is $m$). Let $\nabla f(v)$ denote the gradient of the function $f(x_i, x_{i_1}, \dots, x_{i_k})$ at $v=(v_i, v_{i_1},\dots, v_{i_k})$:
\begin{small}\begin{align}
\nabla f(v) = (\frac{\partial f}{\partial x_i}(v_i), \frac{\partial f}{\partial x_{i_1}}(v_{i_1}),
\ldots, \frac{\partial f}{\partial x_{i_k}}(v_{i_k}))
\end{align}\end{small}
Furthermore, let $|\nabla f(v)|$ denote the sum of magnitude:
\begin{small}\begin{align}
|\nabla f(v)| = 
|\frac{\partial f}{\partial x_i}(v_i)| +
|\frac{\partial f}{\partial x_{i_1}}(v_{i_1})| + \ldots + 
|\frac{\partial f}{\partial x_{i_k}}(v_{i_k})|
\end{align}\end{small}

Assume that $|\nabla f(v)| \leq \delta$ holds for all utility functions in the DCOP and for all $v$. 

\begin{theorem}
The error bound of \emph{discrete DPOP} is $|\setf{F}|m\delta$.
\end{theorem}

\begin{proof}
	First, we prove that the magnitude of the projection of function $f$ is also bounded from above by $\delta$.
	Let $x_i = v_i$ be the point where:
\begin{small}\begin{align}
	g(x_{i_1}, \dots, x_{i_k}) 
	&= f(x_i = v_i, x_{i_1}, \dots, x_{i_k}) \\
	&= \max_{x_i} f(x_i, x_{i_1}, \dots, x_{i_k})
\end{align}\end{small}
	Then, assume that $|\nabla g(v)| > \delta$ for all $v$. Let $v' = (v_i, v_{i_1}, \ldots, v_{i_k})$ and $v'_{-i} = (v_{i_1}, \ldots, v_{i_k})$, then:
\begin{small}\begin{align}
	|\nabla f(v')|
	&= |\frac{\partial f}{\partial x_i}(v_i)| +
	|\frac{\partial f}{\partial x_{i_1}}(v_{i_1})| + \ldots + 
	|\frac{\partial f}{\partial x_{i_k}}(v_{i_k})| \\
	&\ge |\frac{\partial f}{\partial x_{i_1}}(v_{i_1})| + \ldots + 
	|\frac{\partial f}{\partial x_{i_k}}(v_{i_k})| \\
	&= |\nabla g(v'_{-i})| \\
	& > \delta
\end{align}\end{small}
	This contradicts with assumption that $|\nabla f(v)| \le \delta$ for all $v$.
	
	The error bound of each function is then $m \delta$ because each hypercube is of size $m$ and the magnitude of the gradient within each hypercube is at most $\delta$.
	As the error may be accumulated each time an agent sums up utility functions, the total error bound for a problem is thus $|\setf{F}|m\delta$, where $|\setf{F}|$ is the number of utility functions in the problem. 
\end{proof}

\begin{theorem}
The error bound of \emph{AF-DPOP} is $|\setf{F}|(m+|\setf{A}|k\alpha\delta)\delta$, where 
$k$ is the number of times each agent ``moves'' values of its separator by calling Eqs.~\eqref{eq:leaf_move} or~\eqref{eq:nonleaf_move}. 
\end{theorem}
\begin{proof}
After each ``move'' of an agent by calling either Eqs.~\eqref{eq:leaf_move} or~\eqref{eq:nonleaf_move}, the maximum size of the hypercubes increases by $\alpha \delta$, where $\alpha$ is the learning rate. Since each agent performs this update only $k$ times, the largest increase in the size of the hypercube is $k \alpha \delta$. Finally, since the value of an agent can be updated by any of its children or pseudo-children, the total increase in the size of the hypercube is thus $|\setf{A}| k \alpha \delta$, where $|\setf{A}|$ is the number of agents in the problem. Therefore, this combined with the proof of the bound for discrete DPOP, the error bound is thus $|\setf{F}|(m+|\setf{A}|k\alpha\delta)\delta$.
\end{proof}

\begin{theorem}
In a binary constraint graph $G = (\setf{X}, E)$, the number of messages of \emph{HCMS} with $k$ iterations is $4k|E|$. The number of messages of \emph{discrete DPOP}, \emph{AF-DPOP}, and \emph{CAF-DPOP} is $2|\setf{X}|$.
\label{th:num_messages}
\end{theorem}
\begin{proof}
HCMS has the same number of messages as the Max-Sum algorithm \cite{farinelli:08}. Every edge of the constraint graph has two variable nodes and one function node and, thus, it takes 4 messages per edge in one iteration. The total number of messages in HCMS is thus $4k|E|$. 
	
The number of messages required by AF-DPOP and CAF-DPOP is identical to that of DPOP -- each agent sends one UTIL message to its parent and one VALUE message to each of its children in the pseudo-tree. Since pseudo-trees are spanning trees, the number of messages is thus $2|\setf{X}|$.	
\end{proof}

\begin{theorem}
The message size complexity of \emph{discrete DPOP}, \emph{AF-DPOP} and \emph{CAF-DPOP} is $O(d^w)$, $O((d|\setf{X}|)^w)$, and $\max\{|\setf{A}|, k\}$, respectively, where $d$ is the number of points used by each agent to discretize the domain of its separator agents, $w$ is the induced width of the pseudo-tree, and $k$ is the number of clusters used by CAF-DPOP.
\end{theorem}
\begin{proof}
For DPOP, the message size complexity is $O(d^w)$~\cite{petcu:05}. 
For AF-DPOP, as the values of an agent are ``moved'' by their children and pseudo-children, in the worst case, all the values are unique and the maximum number of such values is $O(d|\setf{X}|)$. The message sizes are then similar to discrete DPOP with $O(d|\setf{X}|)$ values per agent. Therefore, its message size complexity is $O((d|\setf{X}|)^w)$.
For CAF-DPOP, the message size complexity of UTIL messages is $O(k)$ since only the utilities of the centroids of $k$ clusters are sent. And the message size complexity of VALUE messages is $O(|\setf{A}|)$, such as in a fully-connected graph where an agent sends the values of every agent from the root of the pseudo-tree down to itself in a VALUE message to its child. Therefore, the message complexity of the algorithm is the $O(\max\{|\setf{A},k\})$.
\end{proof}
%


\section{Experimental Results}
\label{sec:results}

\begin{table*}[t]
	\begin{center}
		\small \centering
		\scalebox{1}
		{%
			\small
			\begin{tabular}{r | r | r | rrrr | rrrr |}
				\cline{2-11}	
				\multicolumn{1}{c|}{\multirow{2}{*}{$|\setf{A}|$}}
				& \multicolumn{1}{c|}{\multirow{2}{*}{HCMS}} & \multicolumn{1}{c|}{\multirow{2}{*}{DPOP}} & \multicolumn{4}{c|}{AF-DPOP} & \multicolumn{4}{c|}{CAF-DPOP} \\
				&  & & 5 & 10 & 15 & 20 & 5 & 10 & 15 & 20 \\
				\cline{2-11}	
				15 & 265k & 522k & 710k & 763k & 824k & 891k &
												639k & 697k & 715k  & 787k \\
				20 & 345k & 865k & 1171k & 1285k & 1334k & 1407k &
												1006k & 1017k & 975k & 973k \\
				25 & 439k & --- & ---  & ---  & ---  & ---  &
												1040k & 1101k & 1024k & 1027k \\
				30 & 506k & --- & ---  & ---  & ---  & ---  &
											 	1513k & 1682k & 1597k & 1656k \\
				\cline{2-11} 
			\end{tabular}%
		}
		\vspace{-0.5em}
		\caption{Experimental Results Varying the Number of Agents on Random Graphs with $p_1 = 0.2$ and Three Initial Discrete Points \label{table:graph-iteration}}
	\end{center}
	\vspace{-0.5em}
\end{table*}

\begin{table*}[t]
\small \center
\begin{tabular}{r | r | r | r | r | r | r | r | r |}
\multicolumn{1}{c}{ } & \multicolumn{3}{c}{\multirow{2}{*}{(a) Random Trees with $|\setf{A}| = 20$}} & 
\multicolumn{1}{c}{ } & \multicolumn{4}{c}{\multirow{2}{*}{(b) Random Graphs with $|\setf{A}| = 20$ and $p_1 = 0.2$}} \\
\multicolumn{9}{c}{ } \\ 
\cline{2-4}
\cline{6-9}
\#points & HCMS & DPOP & AF-DPOP & & HCMS & DPOP & AF-DPOP & CAF-DPOP \\
\cline{2-4}
\cline{6-9}
1 & 0 & 0 & 254k & & 8 & 15k & 428k & 431k \\
3 & 306k & 541k & 870k & & 345k & 865k & 1285k & 1017k \\
9 & 554k & 990k & 1133k & & 706k & --- & --- & 1272k \\
\cline{2-4}
\cline{6-9}
\end{tabular}
\caption{Experimental Results Varying the Number of Discretized Points
\label{table:point}}
\end{table*}




%

We empirically evaluate \emph{EF-DPOP}, \emph{AF-DPOP}, and \emph{CAF-DPOP} against (discrete) \emph{DPOP} and \emph{HCMS} on both random trees and random graphs. We adapt the (discrete) DPOP algorithm to solve F-DCOPs by discretizing the continuous domain into discrete representative points. 

We measure the quality of solutions found by the algorithms as well as the number of messages taken by the algorithm. Since HCMS is an iterative algorithm that may take a long time and a large number of messages before converging, in order for fair comparisons, we initially planned to terminate the algorithm after it sends as many messages as the DPOP-variants. However, even a single iteration of HCMS requires more messages than the DPOP-variants. We thus let HCMS terminate after one iteration. We did not report the actual number of messages since they could be trivially computed via Theorem~\ref{th:num_messages}. 

Tables~\ref{table:tree-iteration} and~\ref{table:graph-iteration} show the various algorithms' solution qualities on random trees and graphs, respectively, where we vary the number of agents $|\setf{A}|$ and every algorithm discretizes the domains of variables into three points. We also vary the number of times AF-DPOP and CAF-DPOP agents ``move'' a point (by calling Eqs.~\eqref{eq:leaf_move} or~\eqref{eq:nonleaf_move}) from 5 to 20. Tables~\ref{table:point}(a) and~\ref{table:point}(b) show the various algorithms' solution qualities on random trees and graphs, respectively, where we set the number of agents $|\setf{A}|$ to 20 and vary the number of initial points used by the algorithms to discretize the domains from 1 to 9. In all our experiments, we set the domain of each agent to be in the range $[-100, 100]$. We generate utility functions that are binary quadratic functions, where the signs and coefficients of the functions are randomly chosen. Our experiments were performed on a 2.10GHz machine with 100GB of RAM.
Results are averaged over 20 runs. 

\paragraph{Random Trees: }

We omit the results of CAF-DPOP from Table~\ref{table:tree-iteration} since it finds identical solutions to AF-DPOP on trees -- there is no need to perform any clustering on trees since an agent does not receive utilities for value combinations of its parent from its children since there are no backedges in the pseudo-tree. 

Not surprisingly, EF-DPOP finds the best solution since it is an exact algorithm. However, it could only solve the smallest of instances -- due to memory limitations, the agents could not store the necessary number of piecewise functions to accurately represent the utility functions after additions and projections. In general, AF-DPOP finds better solutions than DPOP, which finds better solutions than HCMS. The reason is because AF-DPOP updates the value of representative points before propagating up the pseudo-tree. In contrast, the values chosen by DPOP is fixed from the start. Finally, HCMS performs poorly because a single iteration is insufficient for it to converge to a good solution. Additionally, as expected, the quality of solutions found by AF-DPOP improves with increasing number of times points are ``moved'' by the algorithm. 

We omit the results of EF-DPOP from Table~\ref{table:point}(a) as it failed to solve these instances and we omit the results of CAF-DPOP because it finds identical solutions to AF-DPOP on trees. Not surprisingly, the quality of solutions found by all the three algorithms increase with increasing number of points. The reason is that the agents can more accurately represent the utility function with more points.

\paragraph{Random Networks: }

The trends in Table~\ref{table:graph-iteration} are similar to those in random trees, except that CAF-DPOP finds solutions with qualities between that of AF-DPOP and DPOP. The reason is that CAF-DPOP clusters the points into $k$ clusters and only propagate a representative point from each cluster. Therefore, the $k$ points represent the utility functions less accurately than the full number of unclustered points that AF-DPOP uses. However, this reduced number of points propagated also improves the scalability of CAF-DPOP, where it is able to solve problems larger problems than AF-DPOP and DPOP. 

The trends in Table~\ref{table:point}(b) are again similar to that in random trees, except that both AF-DPOP and DPOP ran out of memory with 9 points. Interestingly, CAF-DPOP also finds better solutions than AF-DPOP when they use only 1 point.

\section{Conclusions}
\label{sec:conclusions}


In many real-world applications, agents usually choose their values from continuous ranges. Researchers have thus proposed the F-DCOP formulation to model continuous variables. However, existing methods suffer from the limitation that they do not provide quality guarantees on the solutions found. In this paper, we remedy this limitation by introducing (1)~EF-DPOP, which finds exact solutions on F-DCOPs with linear or quadratic utility functions and are defined over tree-structure graphs; (2)~AF-DPOP, which finds error-bounded solutions on general F-DCOPs; and (3)~CAF-DPOP, which limits the message size of AF-DPOP to a user-defined parameter $k$. Experimental results show that AF- and CAF-DPOP both find better solutions that HCMS, an existing state-of-the-art F-DCOP algorithm, when given the same communication limitations. Therefore, these algorithms combined extend the applicability of DCOPs to more applications that require quality guarantees on the solutions found as well as those that require limited communication capabilities. 

\newpage
\bibliographystyle{named}
\bibliography{f-dcop}

\end{document}